\documentclass[a4paper,twoside,11pt]{article}
\usepackage[margin=2.9cm]{geometry}
\usepackage[affil-it]{authblk}

\usepackage{amssymb,amsmath,amsthm,booktabs}
\usepackage{graphicx,color,hyperref,times} 
\usepackage{enumerate,paralist,xspace}
\usepackage[font=small]{subfig,caption}
\usepackage[ruled,vlined,linesnumbered]{algorithm2e}
\usepackage{todonotes} 

\newcommand{\indeg}[2]{\mathrm{indeg}_{\mathrm{#1}}(#2)}
\newcommand{\outdeg}[2]{\mathrm{outdeg}_{\mathrm{#1}}(#2)}

\newcommand{\nitem}{\medskip \noindent \item}

\begin{document} 
% ============================================================
\title{On the Total Number of Bends for Planar Octilinear Drawings}
\author[1]{Michael~A.~Bekos}
\author[1]{Michael~Kaufmann}
\author[1]{Robert~Krug}
\affil[1]{Wilhelm-Schickhard-Institut f\"ur Informatik, Universit\"at T\"ubingen, Germany\\
$\{$bekos,mk,krug$\}$@informatik.uni-tuebingen.de}
\date{}
\newtheorem{lemma}{Lemma}
\newtheorem{theorem}{Theorem}
\newtheorem{definition}{Definition}
% ============================================================

\maketitle

% =================================================================
\begin{abstract}
An \emph{octilinear drawing} of a planar graph is one in which each
edge is drawn as a sequence of horizontal, vertical and diagonal at
$45^\circ$ line-segments. For such drawings to be readable, special
care is needed in order to keep the number of bends small. As the
problem of finding planar octilinear drawings of minimum number of
bends is NP-hard~\cite{Noellenburg05}, in this paper we focus on
upper and lower bounds. From a recent result of Keszegh et
al.~\cite{KPP13} on the slope number of planar graphs, we can derive
an upper bound of  $4n-10$ bends for $8$-planar graphs with $n$
vertices. We considerably improve this general bound and
corresponding previous ones for triconnected $4$-, $5$- and
$6$-planar graphs. We also derive non-trivial lower bounds for these
three classes of graphs by a technique inspired by the network flow
formulation of Tamassia~\cite{Tamassia87}.
\end{abstract}
%=================================================================

%=================================================================
\section{Motivation and Background}
\label{sec:introduction}
%=================================================================

Octilinear drawings of graphs have a long history of research, which
dates back to the early thirteenth century, when an English technical
draftsman, Henry Charles Beck (also known as Harry Beck), designed
the first schematic map of London Underground. His map, the so-called
Tube map, looked more like an electrical circuit diagram (consisting
of horizontal, vertical and diagonal line segments) rather than a
true map, as the underlying geographic accuracy was neglected. Laying
out networks in such a way is called \emph{octilinear graph drawing}
and plays an important role in map-schematization and the design of
metro-maps. In particular, an octilinear drawing $\Gamma(G)$ of a
graph $G=(V,E)$ is one in which each vertex occupies a point on an
integer grid and each edge is drawn as a sequence of horizontal,
vertical and diagonal at $45^\circ$ line segments. When $G$ is
planar, usually it is required $\Gamma(G)$ to be planar as well.

In planar octilinear graph drawing, an important goal is to keep the
number of bends small, so that the produced drawings can be
understood easily. However, the problem of determining whether a
given embedded planar graph of maximum degree eight admits a
bend-less planar octilinear drawing is
NP-complete~\cite{Noellenburg05}. This motivated us to neglect
optimality and study upper and lower bounds on the total number of
bends of such drawings. Surprisingly enough, very few results were
known, even if the octilinear model has been extensively studied in
the areas of metro-map visualization and map schematization.

One can derive the first (non-trivial) upper bound on the required
number of bends from a result on the planar slope number of graphs
by Keszegh et al.~\cite{KPP13}, who proved that every $k$-planar
graph (that is, planar of maximum degree $k$) has a planar drawing
with at most $\lceil\frac{k}{2}\rceil$ different slopes in which
each edge has at most two bends. For $3 \leq k \leq 8$, the drawings
are octilinear, which yields an upper bound of $6n-12$, where $n$ is
the number of vertices of the graph. The bound can be reduced to
$4n-10$ with some effort; see our subsection on related work.

On the other hand, it is known that every $3$-planar graph with five
or more vertices admits a planar octilinear drawing in which all
edges are bend-less~\cite{Kant92,GLM14}. Also, for $4 \leq k \leq 5$,
it was recently proved that $4$- and $5$-planar graphs admit planar
octilinear drawings with at most one bend per edge~\cite{BGKK14},
which implies that the total number of bends for $4$- and $5$-planar
graphs can be upper bounded by $2n$ and $5n/2$, respectively.

%=================================================================
\begin{table}[t!]
  \centering
  \caption{A short summary of our results.}
  \label{table:results}
  \medskip
  \begin{tabular}{lcccccc}
    \toprule
    & & &\multicolumn{4}{c}{Upper bounds}\\ 
    \cmidrule{4-7}
    Graph class & Lower bound & Ref. & Previous & Ref. & New & Ref.\\
    \midrule 
    $3$-con. $4$-planar & $n/3-1$  & Thm.~\ref{thm:lb} & $2n$    & \cite{BGKK14} & $n+5$   & Thm.~\ref{thm:4ub}\\
    $3$-con. $5$-planar & $2n/3-2$ & Thm.~\ref{thm:lb} & $5n/2$  & \cite{BGKK14} & $2n-2$  & Thm.~\ref{thm:5ub}\\
    $3$-con. $6$-planar & $4n/3-6$ & Thm.~\ref{thm:lb} & $4n-10$ & \cite{KPP13}  & $3n-8$  & Thm.~\ref{thm:6ub}\\
    \bottomrule
  \end{tabular} 
\end{table}
%=================================================================

The remainder of this paper is organized as follows. In
Section~\ref{sec:upperbounds}, we considerably improve all
aforementioned bounds for the classes of triconnected $4$-, $5$- and
$6$-planar graphs. In Section~\ref{sec:lowerbounds}, we present
corresponding lower bounds for these three classes of planar graphs.
We conclude in Section~\ref{sec:conclusions} with open problems and
future work. For a summary of our results also refer to
Table~\ref{table:results}.

% =================================================================
\subsection{Related work.}
\label{sec:relatedwork}
% =================================================================
As already stated, Keszegh et al.~\cite{KPP13} have proved that
every $k$-planar graph admits a planar drawing with at most
$\lceil\frac{k}{2}\rceil$ different slopes in which each edge has at
most two bends. If one appropriately adjusts the slopes of all edge
segments incident to a vertex, then one can show that any $k$-planar
graph, with $3 \leq k \leq 8$, admits a planar octilinear drawing in
which each edge has at most two bends. This implies that any
$k$-planar graph on $n$ vertices can have at most $6n-12$ bends,
where $3 \leq k \leq 8$. One can easily improve this bound to
$4n-10$ as follows. The edge that ``enters'' a vertex from its south
port and the edge that ``leaves'' each vertex from its top port in
the $s$-$t$ ordering of the algorithm of Keszegh et al.~can both be
drawn with only one bend each. Since each vertex is incident to
exactly two such edges (except for the first and last ones in the
$s$-$t$ ordering which are only incident to one such edge each), it
follows that $2n-2$ edges can be drawn with at most one bend. Hence,
the bound of $4n-10$ follows.

Octilinear drawings form a natural extension of the so-called
\emph{orthogonal drawings}, which allow for horizontal and vertical
edge segments only. For such drawings, the bend minimization problem
can be solved efficiently, assuming that the input is an embedded
graph~\cite{Tamassia87}. However, the corresponding minimization
problem over all embeddings of the input graph is
NP-hard~\cite{GT01}. Note that in~\cite{Tamassia87} the author
describes how one can extend his approach, so to compute a
bend-optimal octilinear representation\footnote{Recall that a
representation of a graph describes the angles and the bends of a
drawing, neglecting its exact geometry~\cite{Tamassia87}.} of any
given embedded $8$-planar graph. However, such a representation may
not be realizable by a corresponding planar octilinear
drawing~\cite{BT04}.

For orthogonal drawings, several bounds on the total number of bends
are known. Biedl~\cite{Bie96} presents lower bounds for graphs of
maximum degree $4$ based on their connectivity (simply connected,
biconnected or triconnected), planarity (planar or not) and
simplicity (simple or non-simple with multiedges or selfloops). It
is also known that any $4$-planar graph (except for the octahedron
graph) admits a planar orthogonal drawing with at most two bends per
edge~\cite{BK94,LMS98}. Trivially, this yields an upper bound of
$4n$ bends, which can be improved to $2n+2$~\cite{BK94}. Note that
the best known lower bound is due to Tamassia et al.~\cite{TTV91},
who presented $4$-planar graphs requiring $2n-2$ bends.

Finally, in metro-map visualization many approaches have been
proposed that result in octilinear or nearly-octilinear drawings
(see, e.g.,~\cite{HMN06,Noellenburg05,NW11,SROW11}). However, most
of them are heuristics and therefore do not focus on the
bend-minimization problem explicitly.

% =================================================================
\subsection{Preliminaries.}
\label{sec:preliminaries}
% =================================================================
Central in our approach is the canonical order~\cite{FPP90,Kant92b}
of triconnected planar graphs: Let $G=(V,E)$ be a triconnected
planar graph and let $\Pi = (P_0,\ldots,P_m)$ be a partition of $V$
into paths, such that $P_0 = \{v_1,v_2\}$, $P_m=\{v_n\}$ and $v_2
\rightarrow v_1 \rightarrow v_n$ is a path on the outerface of $G$.
For $k=0,1,\ldots,m$, let $G_k$ be the subgraph induced by
$\cup_{i=0}^k P_i$. Path $P_k$ is called \emph{singleton} if
$|P_k|=1$ and \emph{chain} otherwise;  Partition $\Pi$ is a
\emph{canonical order}~\cite{FPP90,Kant92b} of $G$ if for each
$k=1,\ldots,m-1$ the following hold (see also
Figure~\ref{fig:4p_example}):

\begin{enumerate}[(i)] 
\item $G_k$ is biconnected, 
\item all neighbors of $P_k$ in $G_{k-1}$ are on the outer face of
$G_{k-1}$ and
\item all vertices of $P_k$ have at least one neighbor in $P_j$ for
some $j > k$.
\end{enumerate}

To simplify the description of our algorithms, we direct and color
the edges of~$G$ based on partition~$\Pi$ (similar to Schnyder
colorings~\cite{Fe04}) as follows. The first partition $P_0$ of $\Pi$
defines exclusively one edge (that is, edge $(v_1,v_2)$), which we
color blue and direct towards vertex $v_1$. For each partition
$P_k=\{v_i, \ldots, v_{i+j}\} \in \Pi$ later in the order, let
$v_\ell$ and $v_r$ be the leftmost and rightmost neighbors of $P_k$
in $G_{k-1}$, respectively. In the case where $P_k$ is a chain (that
is, $j > 0$), we color edge $(v_i,v_\ell)$ and all edges between
vertices of $P_k$ blue and direct them towards $v_\ell$. The edge
$(v_{i+j},v_r)$ is colored green and is directed towards $v_r$ (see
Figure~\ref{fig:4p_ref_1}). In the case where $P_k$ is a singleton
(that is, $j = 0$), we color the edges $(v_i,v_\ell)$ and
$(v_{i},v_r)$ blue and green, respectively and we direct them towards
$v_\ell$ and $v_r$. We color the remaining edges incident to $P_k$
towards $G_{k-1}$ (if any) red and we direct them towards $v_i$
(see Figure~\ref{fig:4p_ref_5}).

Given a vertex $v \in V$ of $G$, we denote by $\indeg{x}{v}$
($\outdeg{x}{v}$, respectively) the in-degree (out-degree,
respectively) of vertex $v$ in color $x \in \{r,b,g\}$. Then, it is
not difficult to see that for a vertex $v \in V\setminus \{v_1\}$,
$\outdeg{b}{v}=1$, which implies that the blue subgraph is a
spanning tree of $G$. Similarly, $0 \leq
\outdeg{g}{v},~\outdeg{r}{v} \leq 1$. Hence, the green and the red
subgraphs form two forests of $G$. It also holds that $0 \leq
\outdeg{r}{v} \leq 1$ and $0 \leq
\indeg{b}{v},~\indeg{g}{v},~\indeg{r}{v} \leq d(G)-1$, where $d(G)$
is the degree of $G$. For an example refer to
Figure~\ref{fig:4p_e1}.

\begin{figure}[tb]
    \centering
    \begin{minipage}[b]{.24\textwidth}
        \centering
        \subfloat[\label{fig:4p_e1}{}]
        {\includegraphics[width=\textwidth,page=2]{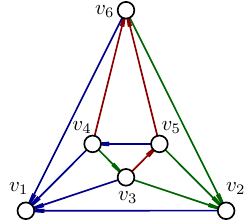}}
    \end{minipage}
    \begin{minipage}[b]{.24\textwidth}
        \centering
        \subfloat[\label{fig:4p_e2}{}]
        {\includegraphics[width=\textwidth,page=3]{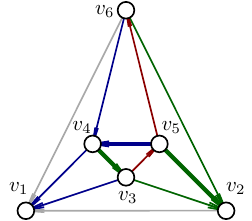}}
    \end{minipage}
    \begin{minipage}[b]{.24\textwidth}
        \centering
        \subfloat[\label{fig:4p_e3}{}]
        {\includegraphics[width=\textwidth,page=4]{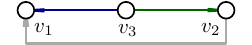}}
    \end{minipage}
    \begin{minipage}[b]{.24\textwidth}
        \centering
        \subfloat[\label{fig:4p_e4}{}]
        {\includegraphics[width=\textwidth,page=5]{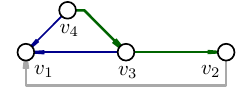}}
    \end{minipage}
    \begin{minipage}[b]{.24\textwidth}
        \centering
        \subfloat[\label{fig:4p_e5}{}]
        {\includegraphics[width=\textwidth,page=6]{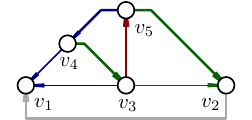}}
    \end{minipage}
    \begin{minipage}[b]{.24\textwidth}
        \centering
        \subfloat[\label{fig:4p_e6}{}]
        {\includegraphics[width=\textwidth,page=7]{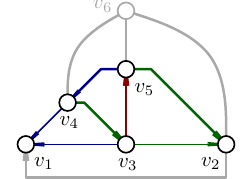}}
    \end{minipage}
    \begin{minipage}[b]{.24\textwidth}
        \centering
        \subfloat[\label{fig:4p_e7}{}]
        {\includegraphics[width=\textwidth,page=8]{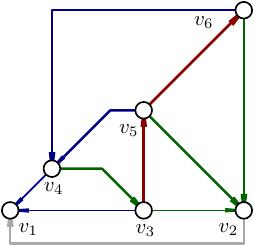}}
    \end{minipage}
    \begin{minipage}[b]{.24\textwidth}
        \centering
        \subfloat[\label{fig:4p_e8}{}]
        {\includegraphics[width=\textwidth,page=9]{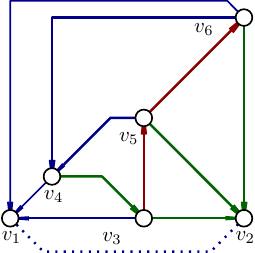}}
    \end{minipage} 
    \caption{
    An illustration of our algorithm for triconnected $4$-planar graphs by an example: the octahedron graph. 
    The underlying canonical order consists of the following partitions:
    $P_0=\{v_1,v_2\}$, $P_1=\{v_3\}$, $P_2=\{v_4\}$, $P_3=\{v_5\}$ and $P_4=\{v_6\}$.
    (a)~The direction and the coloring of the edges. 
    (b)~The corresponding reference edges (bold drawn); the edge $(v_1,v_2)$ of the first partition and the edge $(v_1,v_6)$ incident to the last (degree $4$) partition are ignored. 
    (c)~The placement of the first two partitions.
    (d)~The placement of a singleton of degree $2$ incident to reference edge $(v_4,v_3)$ that is drawn bent.
    (e)~The placement of a singleton of degree $3$ incident to reference edges $(v_5,v_4)$ and $(v_5,v_2)$ that are drawn bent.
    (f)~The last singleton $v_6$ is not incident to reference edges. So, $(v_6,v_4)$, $(v_5,v_6)$ and $(v_6,v_2)$ must be drawn bend-less, which is not possible.
    (g)~Vertex $v_5$ is translated upwards in the direction implied by the edge $(v_3,v_5)$ until one of the horizontal segments incident to $v_5$ is eliminated, which makes the placement of $v_6$ posible.
    (h)~The final layout containing $(v_2,v_1)$ and $(v_6,v_1)$; the dotted edge can be drawn with a single bend.}
    \label{fig:4p_example} 
\end{figure} 

%=================================================================
\section{Upper Bounds}
\label{sec:upperbounds}
%=================================================================

In this section, we present upper bounds on the total number of
bends for the classes of triconnected $4$-planar
(Section~\ref{sec:4planar}), $5$-planar (Section~\ref{sec:5planar})
and $6$-planar (Section~\ref{sec:6planar}) graphs.

%=================================================================
\subsection{Triconnected 4-Planar Graphs.}
\label{sec:4planar}
%=================================================================

Let $G=(V,E)$ be a triconnected $4$-planar graph. Before we proceed
with the description of our approach, we need to define two useful
notions. First, a \emph{vertical cut} is a $y$-monotone continuous
curve that crosses only horizontal segments and divides a drawing
into a left and a right part; see e.g.~\cite{FCK98}. Such a cut
makes a drawing horizontally stretchable in the following sense: One
can shift the right part of the drawing that is defined by the
vertical cut further to the right while keeping the left part of
the drawing in place and the result is a valid octilinear drawing.
Similarly, one can define a \emph{horizontal cut}. 

Since $G$ has at most $2n$ edges, by Euler's formula, it follows that
$G$ has at most $n+2$ faces. In order to construct a drawing
$\Gamma(G)$ of $G$, which has roughly at most $n+2$ bends, we also
need to associate to each face of $G$ a so-called \emph{reference
edge}. This is done as follows. Let $\Pi = \{P_0, \ldots, P_m \}$ be
a canonical order of $G$ and assume that $\Gamma(G)$ is constructed
incrementally by placing a new partition of $\Pi$ each time, so that
the boundary of the drawing constructed so far is a $x$-monotone
path. When placing a new partition $P_k \in \Pi$, $k=1,\ldots,m-1$,
one or two bounded faces of $G$ are formed (note that we treat the
last partition $P_m$ of $\Pi$ separately). More precisely, if $P_k$
is a chain or a singleton of degree $3$ in $G_k$, then only one
bounded face is formed. Otherwise (that is, $P_k$ is a singleton of
degree $4$ in $G_k$), two new bounded faces are formed. In both
cases, each newly-formed bounded  face consists of at least two edges
incident to vertices of $P_k$ and at least one edge of $G_{k-1}$. In
the former case, the reference edge of the newly-formed bounded face,
say $f$, is defined as follows. If $f$ contains at least one green
edge that belongs to $G_{k-1}$, then the reference edge of $f$ is the
leftmost such edge (see Figure~\ref{fig:4p_ref_1} and
\ref{fig:4p_ref_3}). Otherwise, the reference edge of $f$ is the
leftmost blue edge of $f$ that belongs to $G_{k-1}$ (see
Figure~\ref{fig:4p_ref_2} and \ref{fig:4p_ref_4}). In the case where
$P_k$ is a singleton of degree $4$ in $G_k$, the reference edge of
each of the newly formed faces is the edge of $G_{k-1}$ that is
incident to the endpoint of the red edge involved. Observe that by
definition a red edge cannot be a reference edge. For an example see
Figure~\ref{fig:4p_e2}.

\begin{figure}[t]
    \centering 
    \begin{minipage}[b]{.17\textwidth}
        \centering
        \subfloat[\label{fig:4p_ref_1}{}] 
        {\includegraphics[width=\textwidth,page=1]{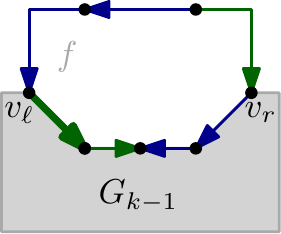}}
    \end{minipage} 
    \begin{minipage}[b]{.17\textwidth}
        \centering 
        \subfloat[\label{fig:4p_ref_2}{}]
        {\includegraphics[width=\textwidth,page=2]{images/4p_ref}}
    \end{minipage}
    \begin{minipage}[b]{.17\textwidth}
        \centering 
        \subfloat[\label{fig:4p_ref_3}{}]
        {\includegraphics[width=\textwidth,page=3]{images/4p_ref}}
    \end{minipage}
    \begin{minipage}[b]{.17\textwidth}
        \centering 
        \subfloat[\label{fig:4p_ref_4}{}]
        {\includegraphics[width=\textwidth,page=4]{images/4p_ref}}
    \end{minipage}
    \begin{minipage}[b]{.17\textwidth}
        \centering 
        \subfloat[\label{fig:4p_ref_5}{}]
        {\includegraphics[width=\textwidth,page=5]{images/4p_ref}}
    \end{minipage}
    \caption{Illustration of the reference edge (bold drawn) in the case of: 
    (a-b) a chain, 
    (c-d) a singleton of degree $2$ in $G_k$ and 
    (e) a singleton of degree $3$ in $G_k$.}
    \label{fig:4p_ref}
\end{figure}

As already stated, we will construct $\Gamma(G)$ in an incremental
manner by placing one partition of $\Pi$ at a time. For the base, we
momentarily neglect the edge $(v_1, v_2)$ of the first partition
$P_0$ of $\Pi$ and we start by placing the second partition, say a
chain $P_1 = \{v_3 , \ldots ,v_{|P_1|+2} \}$, on a horizontal line from left to
right. Since by definition of $\Pi$, $v_3$ and $v_{|P_1|+2}$ are
adjacent to the two vertices, $v_1$ and $v_2$, of the first
partition $P_0$, we place $v_1$ to the left of $v_3$ and $v_2$ to the right of
$v_{|P_1|+2}$. So, they form a single chain where all edges are
drawn using horizontal line-segments that are attached to the east
and west port at their endpoints. The case where $P_1$ is a
singleton is analogous (assuming  that $P_1$ is a chain of unit
length). Assume now that we have already constructed a drawing for
$G_{k-1}$ which has the following invariant properties:

\begin{enumerate}[{IP-}1:] \item \label{ip1}  The number of edges of
$G_{k-1}$ with a bend is at most equal to the number of reference
edges in $G_{k-1}$.
\item \label{ip2} The north-west, north and north-east (south-west,
south and south-east) ports of each vertex are occupied by incoming
(outgoing) blue and green edges and by outgoing (incoming) red
edges\footnote{Note, however, that not all of them can be
simultaneously be occupied due to the degree restriction.}.
\item \label{ip3} If a horizontal port of a vertex is occupied, then
it is occupied either by an edge with a bend (to support vertical
cuts) or by an edge of a chain.
\item \label{ip4} A red edge is not on the outerface of $G_{k-1}$.
\item \label{ip5} A blue (green, respectively) edge of $G_{k-1}$ is
never incident to the north-west (north-east, respectively) port of
a vertex of $G_{k-1}$.
\item \label{ip6} From each reference edge on the outerface of
$G_{k-1}$ one can devise a vertical cut through the drawing of
$G_{k-1}$.
\end{enumerate}
The base of our algorithm conforms with the aforementioned invariant
properties. In the following, we will show how to add the next
partition $P_k$ with $k<m$, so that all invariant properties are
fulfilled. In our description, we will mainly describe the port
assignment at each vertex that will always conform to
IP-\ref{ip2}--\ref{ip5}, which fully specifies how each edge must be
drawn (in other words, we describe the relative coordinates of the
vertices). The exact coordinates can then be computed by adopting an
approach similar to the one of Bekos et al.~\cite{BGKK14}, since the
base of each newly formed face is horizontally stretchable (follows
from IP-\ref{ip6}). Next, we consider the three main cases; see also
Figure~\ref{fig:4p_example} for an example.

\begin{enumerate}[C.1:]

\item \label{c:sing3} \emph{$P_k = \{ v_i \}$ is a singleton of
degree $2$ in $G_k$}; see Figure~\ref{fig:4p_1}, \ref{fig:4p_2}. Let
$v_\ell$ and $v_r$ be the leftmost and rightmost neighbors of $v_i$ in
$G_{k-1}$ (note that $v_\ell$ and $v_r$ are not necessarily
neighboring). We claim that the north-east port of $v_\ell$ and the
north-west port of $v_r$ cannot be simultaneously occupied. For a
proof by contradiction, assume that the claim does not hold. Denote
by $v_\ell \rightsquigarrow v_r$ the path from $v_\ell$ to $v_r$ at
the outerface of $G_{k-1}$ (neglecting the direction of the edges).
By IP-\ref{ip5}, $v_\ell \rightsquigarrow v_r$ starts as blue from
the north-east port of $v_\ell$ and ends as green at the north-west
port of $v_r$. So, inbetween there is a vertex of the path $v_\ell
\rightsquigarrow v_r$ which has a neighbor in $P_j$ for some $j \geq
k$; a contradiction to the degree of $v_i$. Without loss of generality
assume that the north-east port of $v_\ell$ is unoccupied. To draw
the edge $(v_i,v_\ell)$, we distinguish two cases. If $(v_i,v_\ell)$
is the reference edge of a face, then we draw $(v_i,v_\ell)$ as a
horizontal-diagonal combination from the west port of $v_i$ towards
the north-east port of $v_\ell$. Otherwise, $(v_i,v_\ell)$ is drawn
bend-less from the south-west port of $v_i$ towards the north-east 
port of $v_\ell$. To draw the edge $(v_i,v_r)$, again we distinguish
two cases. If the north-west port at $v_r$ is unoccupied, then
$(v_i,v_r)$ will use this port at $v_r$. Otherwise, $(v_i,v_r)$ will
use the north port at $v_r$. In addition, if $(v_i,v_r)$ is the
reference edge of a face, then $(v_i,v_r)$ will use the east port at
$v_i$. Otherwise, the south-east port at $v_i$. The port assignment
described above conforms to IP-\ref{ip2}--\ref{ip5}. Clearly,
IP-\ref{ip1} also holds. IP-\ref{ip6} holds because the newly
introduced edges that are reference edges have a horizontal segment,
which inductively implies that vertical cuts through them are possible.

\begin{figure}[t!]
    \centering 
    \begin{minipage}[b]{.18\textwidth}
        \centering
        \subfloat[\label{fig:4p_1}{}]
        {\includegraphics[width=\textwidth,page=1]{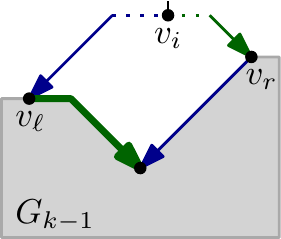}}
    \end{minipage} 
    \begin{minipage}[b]{.18\textwidth}
        \centering 
        \subfloat[\label{fig:4p_2}{}] 
        {\includegraphics[width=\textwidth,page=2]{images/4p}}
    \end{minipage}
    \begin{minipage}[b]{.18\textwidth}
        \centering 
        \subfloat[\label{fig:4p_10}{}]
        {\includegraphics[width=\textwidth,page=10]{images/4p}}
    \end{minipage}
    \begin{minipage}[b]{.18\textwidth}
        \centering 
        \subfloat[\label{fig:4p_3}{}] 
        {\includegraphics[width=\textwidth,page=3]{images/4p}}
    \end{minipage}
    \begin{minipage}[b]{.18\textwidth}
        \centering 
        \subfloat[\label{fig:4p_4}{}] 
        {\includegraphics[width=\textwidth,page=4]{images/4p}}
    \end{minipage}
    \begin{minipage}[b]{.18\textwidth}
        \centering 
        \subfloat[\label{fig:4p_5}{}] 
        {\includegraphics[width=\textwidth,page=5]{images/4p}}
    \end{minipage}
    \begin{minipage}[b]{.18\textwidth}
        \centering 
        \subfloat[\label{fig:4p_6}{}] 
        {\includegraphics[width=\textwidth,page=6]{images/4p}}
    \end{minipage}    
    \begin{minipage}[b]{.18\textwidth}
        \centering 
        \subfloat[\label{fig:4p_7}{}]
        {\includegraphics[width=\textwidth,page=7]{images/4p}}
    \end{minipage}
    \begin{minipage}[b]{.18\textwidth}
        \centering 
        \subfloat[\label{fig:4p_8}{}]
        {\includegraphics[width=\textwidth,page=8]{images/4p}}
    \end{minipage}
    \begin{minipage}[b]{.18\textwidth}
        \centering 
        \subfloat[\label{fig:4p_9}{}]
        {\includegraphics[width=\textwidth,page=9]{images/4p}}
    \end{minipage}
    \caption{Illustration of: 
    (a-b)~the case of a degree-$2$ singleton in $G_k$,
    (c)~the case of a chain,
    (d-j)~the case of a singleton of degree $3$ in $G_k$ 
    (dotted segments can have zero length).}
    \label{fig:4p_part1} 
\end{figure}

\item \label{c:chain} \emph{$P_k=\{v_i,\ldots v_{i+j}\}$ with $j
\geq 1$ is a chain.} This case is similar to case C.\ref{c:sing3},
as $P_k$ has also exactly two neighbors in $G_{k-1}$ (which we again
denote by $v_\ell$ and $v_r$). The edges between $v_i,\ldots
,v_{i+j}$ will be drawn as horizontal segments connecting the west
and east ports of the respective vertices; see
Figure~\ref{fig:4p_10}. The edges $(v_i,v_\ell)$ and $(v_{i+j},v_r)$
are drawn based on the rules of the case C.\ref{c:sing3} (e.g., in
Figure~\ref{fig:4p_10} edge $(v_i,v_\ell)$ is a reference edge, while
the edge $(v_{i+j},v_r)$ is not). Hence, the port assignment still
conforms to IP-\ref{ip2}--IP-\ref{ip5}. In addition, IP-\ref{ip1}
and IP-\ref{ip6} hold, since all edges of the chain are horizontal.

\item \label{c:sing4} \emph{$P_k = \{v_i\}$ is a singleton of
degree~$3$ in $G_k$}. This is the most involved case. Let $v_\ell$
and $v_r$ be the leftmost and rightmost neighbors of $v_i$ in
$G_{k-1}$ and let $v_m$ be the third neighbor of $v_i$ in $G_{k-1}$.
By IP-\ref{ip2} and the degree restriction, the north port of $v_m$
is unoccupied. If the north-east port of $v_\ell$ and the north-west
port of $v_r$ are simultaneously unoccupied, we proceed analogously
to case C.\ref{c:sing3}; see Figure~\ref{fig:4p_3}.
Clearly, IP-\ref{ip1} and IP-\ref{ip6} hold. Consider now the more
involved case, where the north-east port of $v_\ell$ is occupied and
simultaneously $(v_i,v_\ell)$ is not a reference edge. Hence, by
IP-\ref{ip6} $(v_i,v_\ell)$ must be drawn bend-less. Since the
north-east port at $v_\ell$ is occupied, by IP-\ref{ip4} it follows
that the edge at the north-east port of $v_\ell$ is not red.
Therefore, by IP-\ref{ip2} and IP-\ref{ip5}, the edge at the
north-east port of $v_\ell$ is blue. This implies that the path
$v_\ell \rightsquigarrow v_m$ at the outerface of $G_{k-1}$ consists
of exclusively blue edges pointing towards $v_\ell$. Hence, by
IP-\ref{ip5} the north-east port at $v_m$ is unoccupied. Edge
$(v_i,v_\ell)$ can be drawn bend-less if the edge $(v_i,v_r)$ is
a reference edge (that is, by IP-\ref{ip6} $(v_i,v_r)$ has a bend);
see Figure~\ref{fig:4p_4}. In the case where the edge $(v_i,v_r)$ is not
a reference edge (that is, none of $(v_i,v_\ell)$ and $(v_i,v_r)$ is
a reference edge), we need a different argument.
We further distinguish two sub-cases.

\begin{inparaenum}[{C.\ref{c:sing4}}.1:] \nitem \label{c:blue}
\emph{The edge incident to $v_m$ on the path $v_m \rightsquigarrow
v_r$ on the outerface of $G_{k-1}$ is green} (we cope with the case
where this edge is blue later). By definition, the blue (green) edge
of $v_\ell \rightsquigarrow v_m$ ($v_m \rightsquigarrow v_r$) 
incident to $v_m$ is a reference edge and by IP-\ref{ip6} has a
bend. Our aim is to ``eliminate'' one of these bends and draw one of the
edges $(v_i,v_\ell)$ or $(v_i,v_r)$ with a bend and the other one
bend-less. So, IP-\ref{ip1} still holds. In this case, $v_m$ may or
may not be incident to another red edge in $G_{k-1}$ (equivalently,
$v_m$ is either of degree $4$ or $3$, respectively). Without loss of
generality we assume that $v_m$ is incident to another red edge, say
$(v_m,v_m')$, in $G_{k-1}$, that is, $v_m$ is of degree $4$. In this
case, we translate $v_m$ upwards in the direction implied by the
slope of the edge $(v_m,v_m')$, until one of the horizontal segments
of the edges incident to $v_m$ on the outerface of $G_{k-1}$ is
completely eliminated; see Figure~\ref{fig:4p_5}. The only case,
where the aforementioned segment elimination is not possible, is
when $(v_m,v_m')$ is vertical and the edges incident to $v_m$ at the
outerface of $G_{k-1}$ are both horizontal-vertical combinations;
see Figure~\ref{fig:4p_6}. In this particular case, however, by
IP-\ref{ip2} it follows that either the north-west or the north-east
port at $v_m'$ is free. Since both edges incident to $v_m$ at the
outerface of $G_{k-1}$ are bent, by IP-\ref{ip3} we can redraw
$(v_m',v_m)$ so to be diagonal and then we proceed similarly to the
previous case; see Figure~\ref{fig:4p_7}. Also, observe that the
port assignment still conforms to IP-\ref{ip2}--IP-\ref{ip5}.

\nitem \label{c:green} \emph{The edge incident to $v_m$ on the path
$v_m \rightsquigarrow v_r$ on the outerface of $G_{k-1}$ is blue.}
In this case, $v_m$ cannot be incident to another red edge. In the
case where $v_m$ is of degree $3$, we proceed similar to the case
C.\ref{c:sing4}.1, where $v_m$ was of degree $3$. So, we now focus
on the case where $v_m$ is of degree $4$. In this case, the fourth
edge attached to $v_m$ can be either (outgoing) green or (incoming)
blue. In the former case, this edge is a reference edge. In the
latter case, it is part of a chain. In both cases, however, this edge has a
horizontal segment; see Figure~\ref{fig:4p_8}. Hence, we can
translate $v_m$ horizontally to the left so to eliminate the bend of
the edge incident to $v_m$ on the path $v_\ell \rightsquigarrow
v_m$; see Figure~\ref{fig:4p_9}. Note that all invariant properties
are fulfilled once $v_i$ is drawn.
\end{inparaenum}
\end{enumerate}

Note that the coordinates of the newly introduced vertices are
determined by the shape of the edges connecting them to $G_{k-1}$.
If there is not enough space between $v_\ell$ and $v_r$ to
accommodate the new vertices, IP-\ref{ip6} allows us to stretch the
drawing horizontally using the reference edge of the newly formed
face.

To complete the description of our algorithm, it remains to cope
with the last partition $P_m = \{v_n\}$ and describe how to draw the
edge $(v_1,v_2)$ of the first partition $P_0$ of $\Pi$. If $v_n$ is
of degree $3$, we cope with $P_m$ as being an ordinary singleton.
However, if $v_n$ is of degree $4$, then we momentarily ignore the
edge $(v_n,v_1)$ of $P_m$ and proceed to draw the remaining edges
incident to $v_n$, assuming that $P_m$ is again an ordinary
singleton. The edge $(v_n,v_1)$ can be drawn afterwards using two
bends in total. Finally, since by construction $v_1$ and $v_2$ are
horizontally aligned, we can draw the edge $(v_1,v_2)$ with a single
bend, emanating from the south-east port of $v_1$ towards the
south-west port of $v_2$.

\begin{theorem}
Let $G$ be a triconnected $4$-planar graph with $n$ vertices. A
planar octilinear drawing $\Gamma(G)$ of $G$ with at most $n+5$
bends can be computed in $O(n)$ time.
\label{thm:4ub}
\end{theorem}
\begin{proof}
By IP-\ref{ip1}, all bends of $\Gamma(G)$ are in correspondence with
the reference edges of $G$, except for the bends of the edges
$(v_1,v_2)$ and $(v_n,v_1)$. Since the number of reference edges is
at most $n+2$ and the edges $(v_1,v_2)$ and $(v_n,v_1)$ require $3$
additional bends, the total number of bends of $\Gamma(G)$ does not
exceed $n+5$. The linear running time follows from the observation
that we can use the shifting method of Kant~\cite{Kant92} to compute
the actual coordinates of the vertices of $G$, since in the
canonical order the y-coordinates of the vertices that have been
placed at some particular step do not change in subsequent steps
(following a similar approach as in~\cite{BGKK14}).
\end{proof}

%=================================================================
\subsection{Triconnected 5-Planar Graphs.}
\label{sec:5planar}
%=================================================================

Our algorithm for triconnected $5$-planar graphs is an extension of
the corresponding algorithm of Bekos et al.~\cite{BGKK14}, which
computes for a given triconnected $5$-planar graph $G$ on $n$
vertices a planar octilinear drawing $\Gamma(G)$ of $G$ with at most
one bend per edge. Since $G$ cannot have more than $5n/2$ edges, it
follows that the total number of bends of $\Gamma(G)$ is at most
$5n/2$. However, before we proceed with the description of our
extension, we first provide some insights into this algorithm, which
is based on a canonical order $\Pi$ of $G$.
Central are IP-\ref{ip2} and IP-\ref{ip4} of the previous section and
the so-called \emph{stretchability invariant}, according to which all
edges on the outerface of the drawing constructed at some step of the
canonical order have a horizontal segment and therefore one can
devise corresponding vertical cuts to horizontally stretch the
drawing. We claim that we can appropriately modify this algorithm, so
that all red edges of $\Pi$ are bend-less.

Since we seek to draw all red edges of $\Pi$ bend-less, our
modification is limited to singletons. So, let $P_k = \{ v_i \}$ be
a singleton of $\Pi$. The degree restriction implies that $v_i$ has
at most two incoming red edges (we also assume that $P_k$ is not the
last partition of $\Pi$, that is $k \neq m$). We first consider the
case where $v_i$ has exactly one incoming red edge, say
$e=(v_j,v_i)$, with $j<i$. By construction, $e$ must be attached to
one of the northern ports of $v_j$ (that is, north-west, north or
north-east). On the other hand, $e$ can be attached to any of
the southern ports of $v_i$, as $e$ is its only incoming red edge.
This guarantees that $e$ can be drawn bend-less.

Consider now the more involved case, where $v_i$ has exactly two
incoming red edges, say $e=(v_j,v_i)$ and $e' = (v_{j'},v_i)$ and
assume without loss of generality~that $v_j$ is to the left of
$v_{j'}$ in the drawing of $G_{k-1}$. We distinguish three cases
based on the available ports of $v_j$:

\begin{enumerate}[C.1:] 
  
\item \label{c1} \emph{The north-east port of $v_j$ is unoccupied:}
In this case, $e$ emanates from the north-east port of $v_j$ and
leads to the south-west port of $v_i$ (recall that all southern
ports of singleton $v_i$ are dedicated for incoming red edges; in
this case $e$ and $e'$). If the north-west or the north port of
$v_{j'}$ is unoccupied, then $e'$ can be easily drawn bend-less. In
the former case, $e'$ emanates from the north-west port of $v_{j'}$
and leads to the south-east port of $v_i$. In the latter case, $e'$
emanates from the north port of $v_{j'}$ and leads to the south port
of $v_i$. Hence, the aforementioned port assignment fully specifies
the position of $v_i$. It remains to consider the case, where
neither the north-west nor the north port of $v_{j'}$ is unoccupied,
that is, the north-east port of $v_{j'}$ is unoccupied. By our
coloring scheme and IP-\ref{ip2}, $v_{j'}$ has already two incoming
green edges, say $e_g$ and $e_g'$, and $e'$ is the last edge to be
attached at $v_{j'}$; see Figure~\ref{fig:5p_1}. Therefore, there is
no other (bend-less) red edge involved. We proceed by shifting
$v_{j'}$ up in a way that makes all northern ports of $v_{j'}$
unoccupied; see Figure~\ref{fig:5p_2}. Note that we may have to use
a second bend on the outgoing blue edge of $v_{j'}$ (in order to
maintain the stretchability invariant), but on the other hand we can
eliminate one bend from the second green edge $e_g'$; see
Figure~\ref{fig:5p_2}. So, the total number of bends remains
unchanged. In addition, the endpoints of both $e_g$ and $e_g'$ that
are opposite to $v_{j'}$ may have to be moved horizontally to allow
$e_g$ and $e_g'$ to be drawn planar, but by the stretchability
invariant we are guaranteed that this is always possible. Finally,
the stretchability invariant is maintained, since each edge besides
the red ones contains a horizontal segment.

\begin{figure}[t!]
    \centering 
    \begin{minipage}[b]{.18\textwidth}
        \centering
        \subfloat[\label{fig:5p_1}{}]
        {\includegraphics[width=\textwidth,page=1]{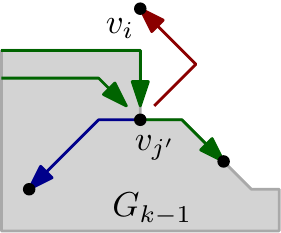}}
    \end{minipage}
    \hfil
    \begin{minipage}[b]{.22\textwidth}
        \centering
        \subfloat[\label{fig:5p_2}{}]
        {\includegraphics[width=\textwidth,page=2]{images/5p}}
    \end{minipage}
    \caption{ 
    (a) $e'=(v_{j'},v_i)$ cannot be drawn bend-less.
    (b) Shifting $v_{j'}$ up resolves the problem.}
    \label{fig:5p} 
\end{figure}

\item \label{c2} \emph{The north-east port of $v_j$ is occupied,
while its north port is unoccupied:} In this case, $e$ emanates from
the north port of $v_j$ and leads to the south port of $v_i$ (that
is, $v_i$ and $v_j$ are vertically aligned). We now claim that the
north-west port of $v_{j'}$ is unoccupied. For the sake of
contradiction, assume that the claim is not true. By our coloring
scheme, the edge attached at the north-west port of $v_{j'}$ is
green, which implies that there must exist a path $v_j
\rightsquigarrow v_{j'}$ at the outerface face of $G_{k-1}$ whose
first edge is blue at the north-east port of $v_j$ and its last edge
is green at the north-west port of $v_{j'}$.
So, path $v_j \rightsquigarrow v_{j'}$ has a vertex which has a
neighbor in $P_\kappa$ for some $\kappa \geq k$. Since $v_i$ is the
only such candidate, the contradiction follows from the degree of
$v_i$. Hence, the north-west port of $v_{j'}$ is unoccupied and
therefore we can draw $e'$ without bends by using the south-east
port of $v_i$ and the north-west port of $v_{j'}$, as desired.

\item \label{c3} \emph{Only the north-west port of $v_j$ is
unoccupied:} We can reduce this case to case C.\ref{c1} by applying
an operation symmetric to the one of Figure~\ref{fig:5p_1}  on vertex
$v_j$. This will result in a configuration where all northern ports
of $v_j$ (including the north-east) are unoccupied.
\end{enumerate}

\begin{theorem}
Let $G$ be a triconnected $5$-planar graph with $n$ vertices. A
planar octilinear drawing $\Gamma(G)$ of $G$ with at most $2n-2$
bends can be computed in $O(n)$ time.
\label{thm:5ub}
\end{theorem}
\begin{proof}
From our extension, it follows that the only edges of $\Gamma(G)$
that have a bend are the blue and the green ones and possibly the
third incoming red edge of vertex $v_n$ of the last partition $P_m$
of $\Pi$. Now, recall that the blue subgraph is a spanning tree of
$G$, while the green one is a forest on the vertices of $G$. So, in
the worst case the green subgraph is a tree on $n-1$ vertices of $G$
(by construction the green subgraph cannot be incident to the first
vertex $v_1$ of $\Pi$). Therefore, at most $2n-2$ edges of
$\Gamma(G)$ have a bend.  In addition, the running time remains
linear since the shifting technique can still be applied. This is
because once a vertex has been placed its $y$-coordinate does not
change anymore, except for the special case of two red edges (cases
C.\ref{c1} and C.\ref{c3}), which does not influence the overall
running time, since it can occur at most once per vertex.
\end{proof}

%=================================================================
\subsection{Triconnected 6-Planar Graphs.}
\label{sec:6planar}
%=================================================================

In this section, we present an algorithm that based on a canonical
order $\Pi =\{ P_0, P_1, \ldots, P_m \}$ of a given triconnected
$6$-planar graph $G=(V,E)$ results in a drawing $\Gamma(G)$ of $G$,
in which each edge has at most two bends. Hence, in total
$\Gamma(G)$ has $6n-12$ bends. Then, we show how one can
appropriately adjust the produced drawing to reduce the total
number of bends.

\begin{figure}[b!]
    \centering 
    \begin{minipage}[b]{.14\textwidth}
        \centering
        \subfloat[\label{fig:6p_1}{}]
        {\includegraphics[width=\textwidth,page=1]{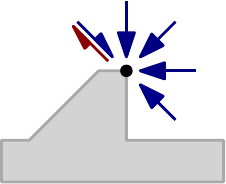}}
    \end{minipage} 
    \hfill
    \begin{minipage}[b]{.14\textwidth}
        \centering
        \subfloat[\label{fig:6p_2}{}]
        {\includegraphics[width=\textwidth,page=2]{images/6p}}
    \end{minipage}  
    \hfill
    \begin{minipage}[b]{.14\textwidth} 
        \centering
        \subfloat[\label{fig:6p_3}{}]
        {\includegraphics[width=\textwidth,page=3]{images/6p}}
    \end{minipage}  
    \hfill
    \begin{minipage}[b]{.14\textwidth}
        \centering
        \subfloat[\label{fig:6p_4}{}]
        {\includegraphics[width=\textwidth,page=4]{images/6p}}
    \end{minipage}  
    \hfill
    \begin{minipage}[b]{.14\textwidth}
        \centering
        \subfloat[\label{fig:6p_5}{}]
        {\includegraphics[width=\textwidth,page=5]{images/6p}}
    \end{minipage}  
    \hfill
    \begin{minipage}[b]{.14\textwidth}
        \centering
        \subfloat[\label{fig:6p_6}{}]
        {\includegraphics[width=\textwidth,page=6]{images/6p}}
    \end{minipage}  
    \hfill
    \begin{minipage}[b]{.18\textwidth}
        \centering
        \subfloat[\label{fig:6p_7}{}]
        {\includegraphics[width=\textwidth,page=7]{images/6p}}
    \end{minipage}  
    \hfill
    \begin{minipage}[b]{.18\textwidth}
        \centering
        \subfloat[\label{fig:6p_8}{}]
        {\includegraphics[width=\textwidth,page=8]{images/6p}}
    \end{minipage}  
    \hfill
    \begin{minipage}[b]{.18\textwidth}
        \centering
        \subfloat[\label{fig:6p_10}{}]
        {\includegraphics[width=\textwidth,page=10]{images/6p}}
    \end{minipage}  
    \hfill
    \begin{minipage}[b]{.18\textwidth}
        \centering
        \subfloat[\label{fig:6p_9}{}]
        {\includegraphics[width=\textwidth,page=9]{images/6p}}
    \end{minipage}  
    \hfill
    \begin{minipage}[b]{.18\textwidth}
        \centering
        \subfloat[\label{fig:6p_11}{}]
        {\includegraphics[width=\textwidth,page=11]{images/6p}}
    \end{minipage}
    \caption{
    (a)-(f)~Illustration of the port assignment computed by Algorithm~\ref{alg:pa}.
    (g)-(k)~Different segment combinations with at most two bends (horizontal segments are drawn dotted)}
    \label{fig:6p_ports}
\end{figure} 

\begin{algorithm}[t!]
\DontPrintSemicolon
\SetKwComment{tcc}{//}{}
\SetKwInOut{Input}{input}
\SetKwInOut{Output}{output}

\Input{A vertex $v$ of a triconnected $6$-planar graph.}
\Output{The port assignment of the edges around $v$, according to the following rules.}
\BlankLine
\begin{enumerate}[{R}1:] 
\item \label{c1} The \texttt{incoming blue edges} of $v$ occupy
consecutive ports in counterclockwise order\\around $v$ starting
from:

\begin{enumerate}[a.]
\item \label{c1a} the south-east port of $v$, if $\indeg{b}{v} + \outdeg{r}{v}=5$; 
see Figure~\ref{fig:6p_1}.
\item \label{c1b} the east port of $v$, if $\indeg{b}{v} +
\outdeg{r}{v}=4$; see Figure~\ref{fig:6p_2}.
\item \label{c1c} the east port of $v$, if $\outdeg{g}{v}=0$ and
(a),(b) do not hold; see Figure~\ref{fig:6p_3}. 
\item \label{c1d} the north-east port of $v$, otherwise;
see Figure~\ref{fig:6p_4}.
\end{enumerate} 

\medskip

\item \label{c2} The \texttt{outgoing red edge} occupies the
counterclockwise next unoccupied port, if $v$ has\\at least
one incoming blue edge. Otherwise, the north-east port of $v$.

\medskip

\item The \label{c3} \texttt{incoming green edges} of $v$ occupy
consecutive ports in clockwise order around $v$ starting from:

\begin{enumerate}[a.]
\item \label{c3a} the west port of $v$, if $\indeg{g}{v} + \outdeg{r}{v} + \indeg{b}{v} \geq 4$; 
see Figure~\ref{fig:6p_5}.
\item \label{c3b} the north-west port of $v$, otherwise; 
see Figure~\ref{fig:6p_6}.
\end{enumerate} 

\medskip

\item \label{c4} The \texttt{outgoing blue edge} of $v$ occupies the
west port of $v$, if it is unoccupied;\\otherwise, the south-west
port of $v$.

\medskip

\item \label{c5} The \texttt{outgoing green edge} of $v$ occupies
the east port of $v$, if it is unoccupied;\\otherwise, the
south-east port of $v$.

\medskip

\item \label{c6} The \texttt{incoming red edges} of $v$ occupy
consecutively in counterclockwise direction the south-west, south
and south-east ports of $v$ starting from the first available.
\end{enumerate}
\caption{\texttt{PortAssignment(v)}} 
\label{alg:pa}
\end{algorithm}

Algorithm~\ref{alg:pa} describes \emph{rules} R\ref{c1} - R\ref{c6}
to assign the edges to the ports of the corresponding vertices. It
is not difficult to see that all port combinations implied by these
rules can be realized with at most two bends, so that all edges have
a horizontal segment (which makes the drawing horizontally
stretchable):
\begin{inparaenum}[(i)]
\item a blue edge emanates from the west or south-west port of a
vertex (by rule R\ref{c4}) and leads to one of the
south-east, east, north-east, north or north-west ports of its
other endvertex (by rule R\ref{c1}); see
Figure~\ref{fig:6p_7} and \ref{fig:6p_8},
\item a green edge emanates from the east or south-east port of a
vertex (by rule R\ref{c5}) and leads to one of the
west, north-west, north or north-east ports of its other endvertex
(by rule R\ref{c3}); see Figure~\ref{fig:6p_10} and
\ref{fig:6p_9},
\item a red edge emanates from one of the north-west, north,
north-east ports of a vertex (by rule R\ref{c2}) and
leads to one of the south-west, south, south-east ports of its
other endvertex (by rule R\ref{c6}); see
Figure~\ref{fig:6p_11}.
\end{inparaenum}

Hence, the shape of each edge is completely determined. To compute
the actual drawing $\Gamma(G)$ of $G$, we follow an incremental
approach according to which one partition (that is, a singleton or a
chain) of $\Pi$ is placed at a time, similar to Kant's
approach~\cite{Kant92b} and the $4$- or $5$-planar case. Each edge is
drawn based on its shape, while the horizontal stretchability ensures
that potential crossings can always be eliminated.  Note additionally
that we adopt the leftist canonical order~\cite{BBC11}, according to
which the leftmost partition is chosen to be placed, when there exist
two or more candidates. Since each edge has at most two bends,
$\Gamma(G)$ has at most $6n-12$ bends in total.

In the following, we reduce the total number of bends. This is done
in two steps. In the first step, we show that all red edges can be
drawn with at most one bend each. Recall that a red edge emanates
from one of the north-west, north, north-east ports of a vertex and
leads to one of the south-west, south, south-east ports of its
other-endvertex. So, in order to prove that all red edges can be
drawn with at most one bend each, we have to consider in total nine
cases, which are illustrated in Figure~\ref{fig:rededges}. It is not
difficult to see that in each of these cases, the red edge can be
drawn with at most one bend. Note that the absence of horizontal
segments at the red edges does not affect the stretchability of
$\Gamma(G)$, since each face of $\Gamma(G)$ has at most two such
edges (which both ``point upward'' at a common vertex). Since a red
edge cannot be incident to the outerface of any intermediate drawing
constructed during the incremental construction of $\Gamma(G)$, it
follows that it is always possible to use only horizontal segments
(of blue and green edges) to define vertical cuts, thus, avoiding all
red edges.

\begin{figure}[t!]
    \centering
    \includegraphics[width=.7\textwidth,page=12]{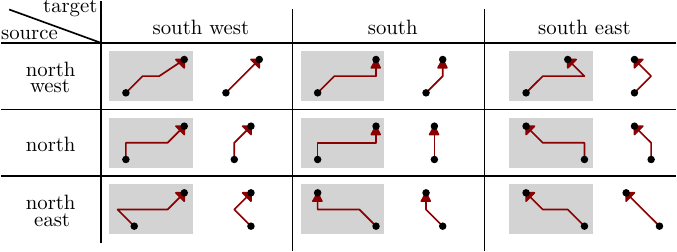} 
    \caption{Red edges can be redrawn with one bend (in gray boxes we show their initial 2-bends shapes)}
    \label{fig:rededges}  
\end{figure} 

The second step of our bend reduction is more involved. Our claim is
that we can ``save'' two bends per vertex\footnote{Except for vertex
$v_1$ of the first partition $P_0$ of $\Pi$, which has no outgoing
blue edge.}, which yields a reduction by roughly $2n$ bends in
total. To prove the claim, consider an arbitrary vertex $u \in V$ of
$G$. Our goal is to prove that there always exist two edges incident
to $u$, which can be drawn with only one bend each. By
rules~R\ref{c3} and R\ref{c4}, it follows that the west port of
vertex $u$ is always occupied, either by an incoming green edge (by
rule~R\ref{c3}) or by a blue outgoing edge (by rule~R\ref{c4}; $u
\neq v_1 \in P_0$). Analogously, the east port of vertex $u$ is
always occupied, either by a blue incoming edge (by rules~R\ref{c1}
and R\ref{c2}) or by an outgoing green edge (by rule~R\ref{c5}). Let
$(u,v) \in E$ be the edge attached at the west port of $u$
(symmetrically we cope with the edge that is attached at the east
port of $u$). If edge $(u,v)$ is attached to a non-horizontal port
at $v$, then $(u,v)$ is by construction drawn with one bend
(regardless of its color; see Figure~\ref{fig:6p_7} and
\ref{fig:6p_10}) and our claim follows.

It remains to consider the case where edge $(u,v)$ is attached to a
horizontal port at $v$. Assume first that edge $(u,v)$ is blue (we
will discuss the case where edge $(u,v)$ is green later). By
Algorithm~\ref{alg:pa}, it follows that edge $(u,v)$ is either the
first blue incoming edge attached at $v$ (by rules R\ref{c1b} and
R\ref{c1c}) or the second one (by rule R\ref{c1a}). We consider each
of these cases separately. In rule R\ref{c1c}, observe that edge
$(u,v)$ is part of a chain (because $\outdeg{g}{u}=0$). Hence, when
placing this chain in the canonical order, we will place $u$
directly to the right of $v$. This implies that $(u,v)$ will be
drawn as a horizontal line segment (that is, bend-less). Similarly,
we cope with rule R\ref{c1b}, when additionally $\outdeg{g}{u}=0$.
So, there are still two cases to consider: rule R\ref{c1a} and rule
R\ref{c1b}, when additionally $\outdeg{g}{u}=1$; see the left part
of Figure~\ref{fig:blueconfiguration}. In both cases, the current
degree of vertex $u$ is $3$ and vertex $v$ (and other vertices that
are potentially horizontally-aligned with $v$) must be shifted
diagonally up, when $u$ is placed based on the canonical order, such
that $(u,v)$ is drawn as a horizontal line segment (that is,
bend-less; see the right part of
Figure~\ref{fig:blueconfiguration}). Note that when $v$ is shifted
up, vertex $v$ and all vertices that are potentially
horizontally-aligned with $v$ are also of degree $3$, since
otherwise one of these vertices would not have a neighbor in some
later partition of $\Pi$, which contradicts the definition of~$\Pi$.

\begin{figure}[t!]
    \centering
    \includegraphics[width=.9\textwidth,page=13]{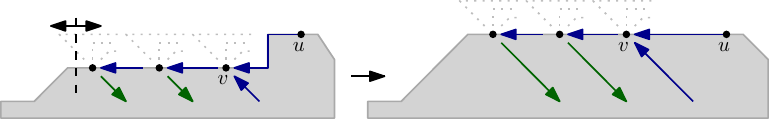}
    \caption{Aligning vertices $u$ and $v$.}
    \label{fig:blueconfiguration}
\end{figure}

We complete our case analysis with the case where edge $(u,v)$ is
green. By rule R\ref{c3a}, it follows that $(u,v)$ is the first
green incoming edge attached at $u$. In addition, when $(u,v)$ is
placed based on the canonical order, there is no red outgoing edge
attached at $u$ (otherwise $u$ would not be at the outerface of the
drawing constructed so far). The leftist canonical order also
ensures that there is no blue incoming edge at $u$ drawn before
$(u,v)$. Hence, vertex $u$ is of degree two, when edge $(u,v)$ is
placed. Hence, it can be shifted up (potentially with other vertices
that are horizontally-aligned with $u$), such that $(u,v)$ is drawn
as a horizontal line segment (that is, bend-less). We summarize our
approach in the following theorem.

\begin{theorem}
Let $G$ be a triconnected $6$-planar graph with $n$ vertices. A
planar octilinear drawing $\Gamma(G)$ of $G$ with at most $3n-8$
bends can be computed in $O(n^2)$ time.
\label{thm:6ub} 
\end{theorem}
\begin{proof}
Before the two bend-reduction steps, $\Gamma(G)$ contains at most
$6n-12$ bends. In the first reduction step, all red edges are drawn
with one bend. Hence, $\Gamma(G)$ contains at most $5n-9$ bends. In
the second reduction step, we ``save'' two bends per vertex (except
for $v_1 \in P_0$, which has no outgoing blue edge), which yields a
reduction by $2n-1$ bends. Therefore, $\Gamma(G)$ contains at most
$3n-8$ bends in total. On the negative side, we cannot keep the
running time of our algorithm linear. The reason is the second
reduction step, which yields changes in the $y$-coordinates of the
vertices. In the worst case, however, quadratic time suffices.
\end{proof}

Note that there exist $6$-planar graphs that do not admit planar
octilinear drawings with at most one bend per edge~\cite{BGKK14}.
Theorem~\ref{thm:6ub} implies that on average one bend per edge
suffices.

%=================================================================
\section{Lower Bounds}
\label{sec:lowerbounds}
%=================================================================

In this section, we present lower bounds on the total number of
bends for the classes of triconected $4$-, $5$- and $6$-planar
graphs.

%=================================================================
\subsection{4-Planar Graphs.} 
\label{sec:4planarl}
%=================================================================

We start our study with the case of $4$-planar graphs.
Our main observation is that if a $3$-cycle $\mathcal{C}_3$ of a
graph has at least two vertices, with at least one neighbor in the
interior of $\mathcal{C}_3$ each, then at least one edge of
$\mathcal{C}_3$ must contain a bend, since the sum of the interior
angles at the corners of $\mathcal{C}_3$ exceeds $180^\circ$. In
fact, elementary geometry implies that a $k$-cycle, say
$\mathcal{C}_k$ with $k \geq 3$, whose vertices have $\sigma \geq 0$
neighbors in the interior of $\mathcal{C}_k$ requires (at least)
$\max \{0, \lceil (\sigma -3k + 8)/3 \rceil \}$ bends.
Therefore, a bend is necessary. Now, refer to the $4$-planar graph
of Figure~\ref{fig:lower_bound_4_planar}, which contains $n/3$ nested
triangles, where $n$ is the number of its vertices. It follows that
this particular graph requires at least $n/3-1$ bends in total.
% when drawn in the octilinear drawing model (the innermost triangle
% can be drawn bendless).

\begin{figure}[t!]
    \centering
    \begin{minipage}[b]{.32\textwidth}
        \centering
        \subfloat[\label{fig:lower_bound_4_planar}{A $4$-planar graph}] 
        {\includegraphics[width=.95\textwidth,page=1]{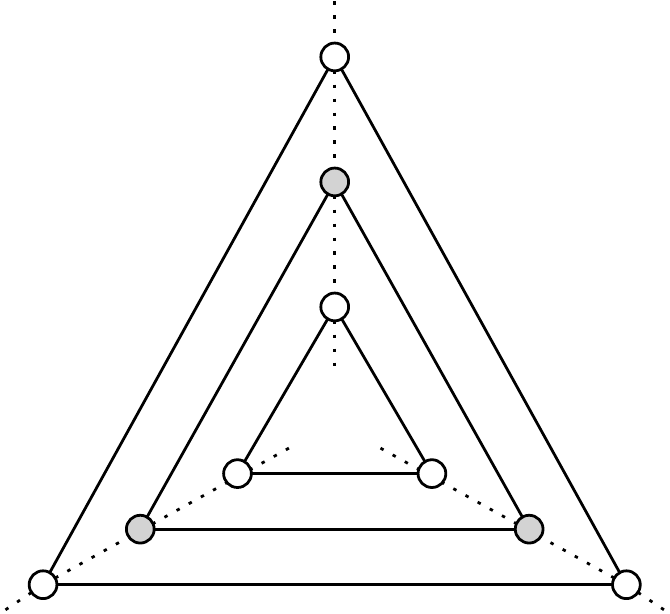}}
    \end{minipage}
    \begin{minipage}[b]{.32\textwidth}
        \centering
        \subfloat[\label{fig:lower_bound_5_planar}{A $5$-planar graph}]
        {\includegraphics[width=.95\textwidth,page=2]{images/lower_bounds}}
    \end{minipage} 
    \begin{minipage}[b]{.32\textwidth}
        \centering
        \subfloat[\label{fig:lower_bound_6_planar}{A $6$-planar graph}]
        {\includegraphics[width=.95\textwidth,page=3]{images/lower_bounds}}
    \end{minipage}
    \caption{
    Planar graphs of different degrees that require
    (a)~$n/3-1$, (b)~$2n/3-2$ and (c)~$4n/3-6$ bends.} 
    \label{fig:lower_bounds}
\end{figure}

%=================================================================
\subsection{5- and 6-Planar Graphs.} 
\label{sec:56planarl}
%=================================================================

For $5$- and $6$-planar graphs, our proof becomes more complex. For
these classes of graphs, we follow an approach inspired by Tamassia's
min-cost flow formulation~\cite{Tamassia87} for computing
bend-minimum representations of embedded planar graphs of bounded
degree. Since it is rather difficult to implement this algorithm in
the case where the underlying drawing model is not the orthogonal
model, we developed an ILP instead. Recall that a representation
describes the ``shape'' of a drawing without specifying its exact
geometry. This is enough to determine a lower bound on the number of
bends, even if a bend-optimal octilinear representation may not be
realizable by a corresponding (planar) octilinear drawing.

In our formulation, variable $\alpha(u,v) \cdot 45^\circ$
corresponds to the angle formed at vertex $u$ by edge $(u,v)$
and its cyclic predecessor around vertex $u$. Hence, $1 \leq
\alpha(u,v) \leq 8$. Since the sum of the angles around a vertex
is $360^\circ$, it follows that $\sum_{v \in N(u)}a(u,v) =
8$.  Given an edge $e=(u,v)$, variables $\ell_{45}(u,v)$,
$\ell_{90}(u,v)$ and $\ell_{135}(u,v)$ correspond to the number of
left turns at $45^\circ$, $90^\circ$ and $135^\circ$ when moving
along $(u,v)$ from vertex $u$ towards vertex $v$. Similarly,
variables $r_{45}(u,v)$, $r_{90}(u,v)$ and $r_{135}(u,v)$ are
defined for right turns. All aforementioned variables are integer
lower-bounded by zero. For a face $f$, we assume that
its edges are directed according to the clockwise traversal of $f$.
This implies that each (undirected) edge of the graph appears twice
in our formulation.  For reasons of symmetry, we require
$\ell_{45}(u,v) = r_{45}(v,u)$, $\ell_{90}(u,v) = r_{90}(v,u)$ and
$\ell_{135}(u,v) = r_{135}(v,u)$. Since the sum of the angles formed
at the vertices and at the bends of a bounded face $f$ equals to
$180^\circ\cdot(p(f)-2)$, where $p(f)$ denotes the total number of
such angles, it follows that $\sum_{(u,v) \in
E(f)}\alpha(u,v)+(\ell_{45}(u,v)+\ell_{90}(u,v)+\ell_{135}(u,v)) -
(r_{45}(u,v)+r_{90}(u,v)+r_{135}(u,v))=4a(f)-8$, where $a(f)$
denotes the total number of vertex angles in $f$, and, $E(f)$ the
directed arcs of $f$ in its clockwise traversal. If $f$ is
unbounded, the respective sum is increased by $16$. Of course, the
objective is to minimize the total number of bends over all edges,
or, equivalently $\min \sum_{(u,v)\in E}\ell_{45}(u,v)
+ \ell_{90}(u,v) + \ell_{135}(u,v) + r_{45}(u,v)+r_{90}(u,v)
+ r_{135}(u,v)$.

Now, consider the $5$-planar graph of
Figure~\ref{fig:lower_bound_5_planar} and observe that each ``layer''
of this graph consist of six vertices that form an octahedron
(solid-drawn), while octahedrons of consecutive layers are connected
with three edges (dotted-drawn). Using our ILP formulation, we prove
that each octahedron subgraph requires at least $4$ bends, when drawn
in the octilinear model (except for the innermost one for which we
can guarantee only two bends). This implies that $2n/3-2$ bends are
required in total to draw the graph of
Figure~\ref{fig:lower_bound_5_planar}. For the $6$-planar case, we
apply our ILP approach to a similar graph consisting of nested
octahedrons that are connected by six edges each; see
Figure~\ref{fig:lower_bound_6_planar}. This leads to a better lower
bound of $4n/3-6$ bends, as each octahedron except for the innermost
one requires $8$ bends. Summarizing we have:

\begin{theorem}
There exists a class $G_{n,k}$ of triconnected embedded $k$-planar
graphs, with $4 \leq k \leq 6$, whose octilinear drawings
require at least:
\begin{inparaenum}[(i)]
  \item $n/3-1$ bends, if $k=4$,
  \item $2n/3-2$ bends, if $k=5$ and
  \item $4n/3-6$ bends, if $k=6$.
\end{inparaenum}
\label{thm:lb}
\end{theorem}

%=================================================================
\section{Conclusions}
\label{sec:conclusions}
%=================================================================

In this paper, we studied bounds on the total number of bends of octilinear
drawings of triconnected planar graphs. We showed how one can adjust an
algorithm of Keszegh et al.~\cite{KPP13} to derive an upper bound of $4n-10$
bends for general $8$-planar graphs.
Then, we improved this general bound and previously-known ones for the classes
of triconnected $4$-, $5$- and $6$-planar graphs. For these classes of graphs,
we also presented corresponding lower bounds.\\
We mention two major open problems in this context. The first one is to extend
our results to biconnected and simply connected graphs and to further tighten
the bounds. Since our drawing algorithms might require super-polynomial area
(cf.~arguments from~\cite{BGKK14}), the second problem is to study trade-offs
between the total number of bends and the required area.

\bibliographystyle{abbrv}
\bibliography{references}

\begin{thebibliography}{10}

\bibitem{BBC11}
M.~Badent, U.~Brandes, and S.~Cornelsen.
\newblock More canonical ordering.
\newblock {\em Journal of Graph Algorithms and Applications}, 15(1):97--126,
  2011.

\bibitem{BGKK14}
M.~A. Bekos, M.~Gronemann, M.~Kaufmann, and R.~Krug.
\newblock Planar octilinear drawings with one bend per edge.
\newblock {\em Journal of Graph Algorithms and Applications}, 19(2):657--680,
  2015.

\bibitem{Bie96}
T.~C. Biedl.
\newblock New lower bounds for orthogonal graph drawings.
\newblock In F.~J. Brandenburg, editor, {\em Graph Drawing}, volume 1027 of
  {\em LNCS}, pages 28--39. Springer, 1996.

\bibitem{BK94}
T.~C. Biedl and G.~Kant.
\newblock A better heuristic for orthogonal graph drawings.
\newblock {\em Computational Geometry}, 9(3):159--180, 1998.

\bibitem{BT04}
H.~L. Bodlaender and G.~Tel.
\newblock A note on rectilinearity and angular resolution.
\newblock {\em Journal of Graph Algorithms and Applications}, 8(1):89--94,
  2004.

\bibitem{FPP90}
H.~De~Fraysseix, J.~Pach, and R.~Pollack.
\newblock How to draw a planar graph on a grid.
\newblock {\em Combinatorica}, 10(1):41--51, 1990.

\bibitem{GLM14}
E.~{Di Giacomo}, G.~Liotta, and F.~Montecchiani.
\newblock The planar slope number of subcubic graphs.
\newblock In A.~Pardo and A.~Viola, editors, {\em LATIN}, volume 8392 of {\em
  LNCS}, pages 132--143. Springer, 2014.

\bibitem{Fe04}
S.~Felsner.
\newblock Schnyder woods or how to draw a planar graph?
\newblock In {\em Geometric Graphs and Arrangements}, Advanced Lectures in
  Mathematics, pages 17--42. Vieweg/Teubner Verlag, 2004.

\bibitem{FCK98}
U.~{F\"ossmeier}, C.~Hess, and M.~Kaufmann.
\newblock On improving orthogonal drawings: The {4M}-algorithm.
\newblock In S.~Whitesides, editor, {\em Graph Drawing}, volume 1547 of {\em
  LNCS}, pages 125--137. Springer, 1998.

\bibitem{GT01}
A.~Garg and R.~Tamassia.
\newblock On the computational complexity of upward and rectilinear planarity
  testing.
\newblock {\em SIAM Journal on Computing}, 31(2):601--625, 2001.

\bibitem{HMN06}
S.-H. Hong, D.~Merrick, and H.~A.~D. {do Nascimento}.
\newblock Automatic visualisation of metro maps.
\newblock {\em Journal of Visual Languages and Computing}, 17(3):203--224,
  2006.

\bibitem{Kant92b}
G.~Kant.
\newblock Drawing planar graphs using the lmc-ordering.
\newblock In {\em FOCS}, pages 101--110. IEEE, 1992.

\bibitem{Kant92}
G.~Kant.
\newblock Hexagonal grid drawings.
\newblock In E.~W. Mayr, editor, {\em WG}, volume 657 of {\em LNCS}, pages
  263--276. Springer, 1992.

\bibitem{KPP13}
B.~Keszegh, J.~Pach, and D.~P{\'a}lv{\"o}lgyi.
\newblock Drawing planar graphs of bounded degree with few slopes.
\newblock {\em SIAM Journal of Discrete Mathematics}, 27(2):1171--1183, 2013.

\bibitem{LMS98}
Y.~Liu, A.~Morgana, and B.~Simeone.
\newblock A linear algorithm for 2-bend embeddings of planar graphs in the
  two-dimensional grid.
\newblock {\em Discrete Applied Mathematics}, 81(1-3):69--91, 1998.

\bibitem{Noellenburg05}
M.~N\"ollenburg.
\newblock Automated drawings of metro maps.
\newblock Technical Report 2005-25, Fakult\"at f\"ur Informatik, Universit\"at
  Karlsruhe, 2005.

\bibitem{NW11}
M.~N{\"o}llenburg and A.~Wolff.
\newblock Drawing and labeling high-quality metro maps by mixed-integer
  programming.
\newblock {\em IEEE Transactions on Visualization and Computer Graphics},
  17(5):626--641, 2011.

\bibitem{SROW11}
J.~M. Stott, P.~Rodgers, J.~C. Martinez-Ovando, and S.~G. Walker.
\newblock Automatic metro map layout using multicriteria optimization.
\newblock {\em IEEE Transactions on Visualization and Computer Graphics},
  17(1):101--114, 2011.

\bibitem{Tamassia87}
R.~Tamassia.
\newblock On embedding a graph in the grid with the minimum number of bends.
\newblock {\em SIAM Journal of Computing}, 16(3):421--444, 1987.

\bibitem{TTV91}
R.~Tamassia, I.~G. Tollis, and J.~S. Vitter.
\newblock Lower bounds for planar orthogonal drawings of graphs.
\newblock {\em Information Processing Letters}, 39(1):35--40, 1991.

\end{thebibliography}

\end{document}